\documentclass[lettersize,journal]{IEEEtran}
\usepackage[numbers,sort&compress]{natbib}
\usepackage{multirow}
\usepackage{threeparttable}
\usepackage{amsmath,amsfonts,amssymb}
\usepackage{graphicx}
\usepackage{subfig}
\usepackage{algorithmicx}
\usepackage[ruled,linesnumbered]{algorithm2e}
\usepackage{algpseudocode}
\usepackage{textcomp}
\usepackage{makecell}
\usepackage{multirow,multicol}
\usepackage[colorlinks]{hyperref}
\usepackage[table]{xcolor}
\usepackage{threeparttable}
\usepackage{amsthm}
\newtheorem{theorem}{Theorem}
\begin{document}
	
	
	\title{Sparse Code Transceiver Design \\for Unsourced Random Access with Analytical Power Division in Gaussian MAC}
	
	\author{Zhentian Zhang, Mohammad Javad Ahmadi, Jian Dang, Kai-Kit Wong, Zaichen Zhang, Christos Masouros
		\vspace{-5mm}
		
		\thanks{ }
		\thanks{Zhentian Zhang, J. Dang, Zaichen Zhang are with the National Mobile Communications Research Laboratory, Frontiers Science Center for Mobile Information Communication and Security, Southeast University, Nanjing, 210096, China. Jian Dang, Zaichen Zhang are also with the Purple Mountain Laboratories, Nanjing 211111, China (e-mail: \{zhangzhentian, dangjian, zczhang\}@seu.edu.cn).}
		\thanks{M. J. Ahmadi is with the Chair of Information Theory and Machine Learning, Technische Universität Dresden, 01062 Dresden, Germany (e-mail: Mohammad\_javad.ahmadi@tu-dresden.de).}
		\thanks{K. K. Wong and  C. Masouros are with the Department of Electronic and Electrical Engineering, University College London, Torrington Place, WC1E 7JE, United Kingdom and K. K. Wong is also with the Yonsei Frontier Lab., Yonsei University, 03722 Korea. (e-mails: \{kai-kit.wong, c.masouros\} @ucl.ac.uk).}
		\thanks{M. J. Ahmadi and Zhentian Zhang contribute equally to this work.}
		\thanks{Corresponding authors: Jian Dang (dangjian@seu.edu.cn).}}
	%
	%
	
	\maketitle
	\begin{abstract}
		In this work, we discuss the problem of unsourced random access (URA) over a Gaussian multiple access channel (GMAC). To address the challenges posed by emerging massive machine-type connectivity, URA reframes multiple access as a coding-theoretic problem. The sparse code-oriented schemes are highly valued because they are widely used in existing protocols, making their implementation require only minimal changes to current networks. However, drawbacks such as the heavy reliance on extrinsic feedback from powerful channel codes and the lack of transmission robustness pose obstacles to the development of sparse codes. To address these drawbacks, a novel sparse code structure based on a universally applicable power division strategy is proposed. Comprehensive numerical results validate the effectiveness of the proposed scheme. Specifically, by employing the proposed power division method, which is derived analytically and does not require extensive simulations, a performance improvement of approximately 2.8 dB is achieved compared to schemes with identical channel code setups.
	\end{abstract}
	\begin{IEEEkeywords}
		Unsourced random access, sparse code, TIN-SIC power division, Gaussian MAC.
	\end{IEEEkeywords}
	
	\section{Introduction}
	\subsubsection{Background and Related Work}
To address the challenges of emerging massive connectivity and support future massive machine-type communication (mMTC) networks, unsourced random access (URA) redefines the multiple access design as a coding-theoretical problem~\cite{Polyanskiy}. In general, URA aims to develop a code that can accommodate a large number of user equipments (UEs) simultaneously, while operating under finite block lengths and limited channel uses. The potential advantages in energy and spectral efficiency, along with its low complexity and rapid access features, have drawn significant attention from researchers in academia \cite{acdemia and industrial1} and industry \cite{acdemia and industrial2}.
	
	The Gaussian multiple access channel (GMAC) is a key model for investigating practical URA code design. Several multiple access code structures have been proposed to approach the achievable limit of GMAC-URA obtained in \cite{Polyanskiy}. Specifically, the random spreading (RS) code \cite{PD_Mohamod,RS-Polar,RS-LDPC,SKP} leverages spectral spreading gain to reduce multi-user interference and achieve improved estimation precision. Moreover, the coded/coupled compressive sensing (CCS) category \cite{CCS,CCS2,CRC-BMST} performs segmented transmission using a tree-based decoder, resulting in overall complexity that scales linearly with activity. Especially, sparse code-oriented schemes \cite{ODMA_RA,ODMA_NB_LDPC,ODMA_Polar} require minimal modification to existing communication systems, and their scalability and flexibility make them strong candidates for URA code design.
    \subsubsection{Motivations}
    While sparse code structures offer many advantageous features, several issues remain to be addressed. One challenge is that the performance of current implementations heavily relies on the extrinsic feedback from the channel decoder. This leads to significant variation in performance, depending on the channel code used. For example, consider on-off multiple access (ODMA) \cite{ODMA_RA}, a widely used sparse code structure that can be deployed with various channel codes. ODMA with low-density parity-check (LDPC) \cite{ODMA_Polar} demonstrates nearly 3 dB worse capacity performance compared to ODMA with non-binary LDPC \cite{ODMA_NB_LDPC} at 300 active UEs. As the design of powerful channel codes becomes increasingly complex, the viability of URA sparse codes is put into question. Another challenge is the robustness of the current transmission structure, which requires further improvement. For instance, ODMA with Polar codes reaches the theoretical limit with fewer than 100 UEs, but the performance gap widens significantly as the number of UEs increases. To address this, system design must evolve to better accommodate the growing user densities.
	
	\subsubsection{Contributions}
	This paper focuses on the design of sparse codes for GMAC URA, with the development of a power division (PD) strategy aimed at enhancing system performance. Unlike the approach in \cite{PD_Mohamod}, the proposed PD eliminates the need for exhaustive search and provides power allocation ratios through theoretical solutions, formulated based on the scenario of treating interference as noise with successive interference cancellation (TIN-SIC)~\cite{andreev2020polar, Ahmadi2023Unsourced}. Furthermore, the slotted structure limits the number of simultaneously active users, reducing the load and computational complexity of the multi-user detector (MUD) while improving detection reliability.

        \subsubsection{Content Structures}
	In Sec.~\ref{GMAC}, the transceiver design of the proposed sparse code for URA is elaborated. In Sec.\ref{TIN-SIC}, the proposed TIN-SIC PD strategy is explained by detailed mathematics. Subsequently, the numerical results are illustrated in Sec.~\ref{numerical} to validate the superiority of the proposed TIN-SIC PD method and the proposed sparse code structure. Eventually, conclusion is stated in Sec.~\ref{conclusion}.
	\begin{figure}[htp]
		\centering
		\includegraphics[width=3.5in]{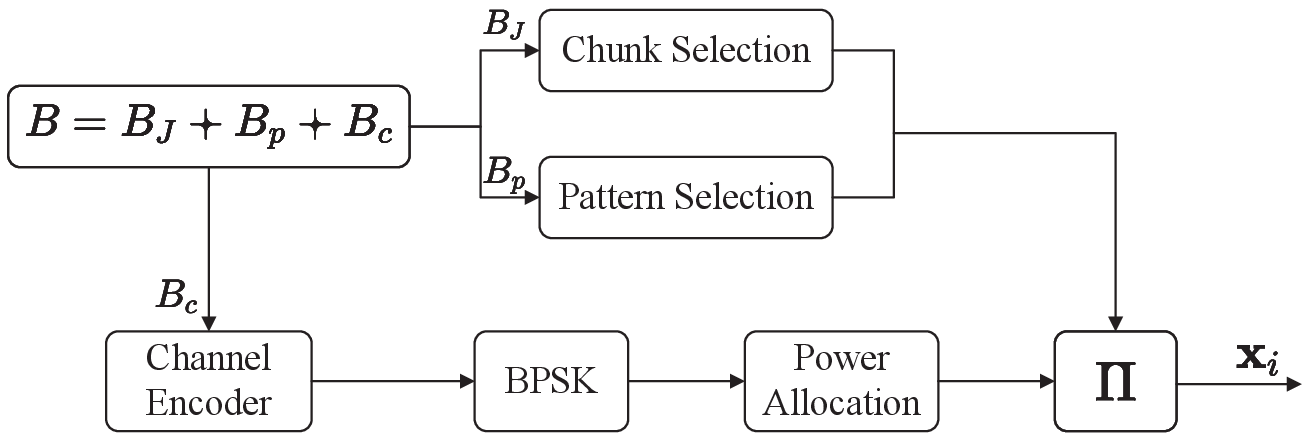}
		\caption{Illustrations of encoding procedures under GMAC.}\label{fig:GMAC_Encoder}
	\end{figure}
	\section{Transceiver Design}\label{GMAC}
In this section, the problem of multiple access within the context of URA is addressed. The details of the proposed basic URA code are elaborated under GMAC, where a receiver with a single antenna and unit channel gain is considered. The receiver decodes the overlapped noisy observations with iterative Gaussian approximation decoder. Meanwhile, power division method is utilized as an auxiliary approach to enhance the decoding performance.
	
	\begin{figure*}[htp]
		\centering
		\includegraphics[width=6in]{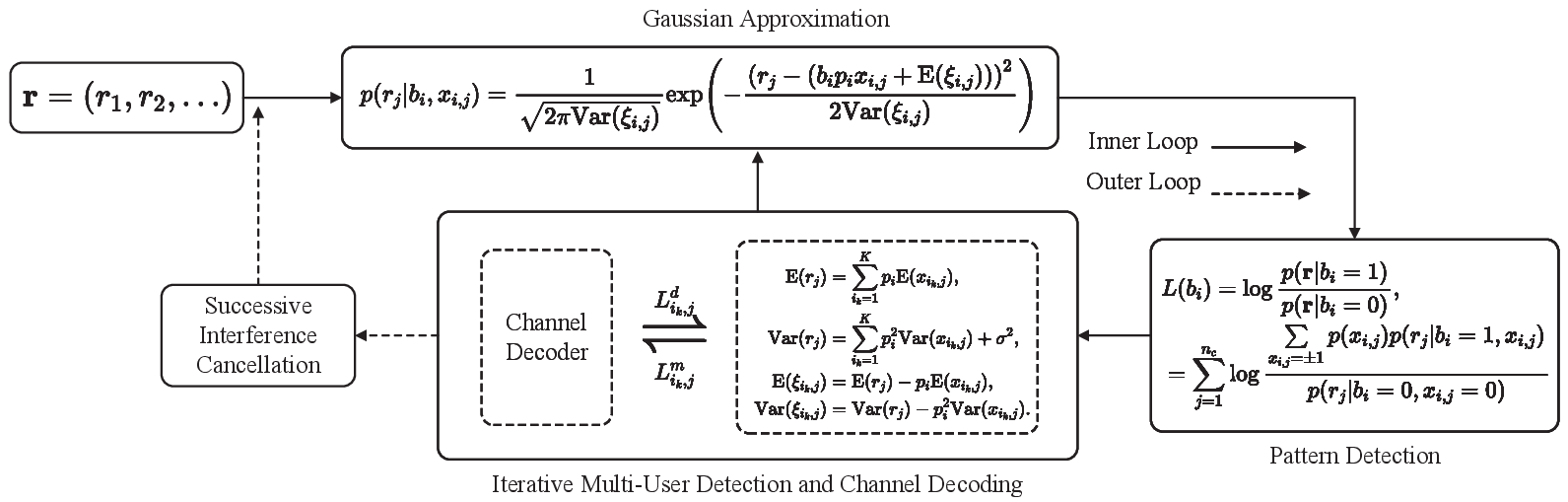}
		\caption{{Illustration of iterative decoding of the proposed GMAC scheme}: The \textit{inner loop} passes statistic informaiton to conduct joint data and pattern detection; The \textit{outer loop} conducts SIC to reduce interference.}\label{fig:GMAC_Decoder}
	\end{figure*}
	\subsubsection{Signal Model and Encoder Design}
There are $K_a$ active UEs sharing finite block length with $n$ channel uses. Each UE intends to transmit $B$ bits of information. The duration of the transmission is uniformly divided into \(J\) chunks, each with \(n_J\) channel uses, such that \(n = n_J J\). The average number of users in a single chunk is denoted as \(\bar{K}_a = \frac{K_a}{J}\). Each UE divides its \( B \) bits of information into three parts with lengths \( B_J \), \( B_p \), and \( B_c \), where \( B = B_p + B_J + B_c \). The first part, consisting of \( B_J \) bits, determines which chunk to transmit the signal in, where \( B_J = \log_2 J \). The second part, of length \( B_p \), is mapped to a column of the pattern codebook \( \mathcal{P} \in \{0,1\}^{n_p\times 2^{B_p}} \) to select a permutation pattern. Each column of this codebook, denoted as \( \mathbf{p}_i \in \{0,1\}^{n_p\times 1} \) for \( i \in \{1,2,\ldots, 2^{B_p}\} \), represents a unique pattern with exactly \( n_c \) nonzero elements, while the remaining \( n_p - n_c \) elements are zeros. The third part, of length \( B_c \), is encoded using a channel code, modulated with binary phase shift keying (BPSK), permuted according to the selected permutation pattern, and zero-padded to construct the transmitted signal \( \mathbf{x}_i \in \{-1,0,1\}^{n_p \times 1} \), which is then transmitted through the selected chunk. The encoding procedures of singe UE is depicted in Fig.~\ref{fig:GMAC_Encoder}. Ignoring asynchronous errors, the noisy received signal within a single chunk can be written as:
	\begin{equation}\label{eq:1}
		\mathbf{r} = \sum_{i=1}^{2^{B_p}}b_ip_i\mathbf{x}_i+\mathbf{n},
	\end{equation}   
 where \(\mathbf{n}\) represents the additive white Gaussian noise (AWGN) with zero mean and variance \(\sigma^2\), and \(p_i\) is the power division ratio, which is a positive constant, with \(p_i = 1\) when no power division is considered. Additionally, \(b_i \in \{0, 1\}\) is the activity indicator; if \(b_i = 1\), the pattern \(\mathbf{p}_i\) is considered active, and the receiver attempts to restore \(\mathbf{x}_i\); if \(b_i = 0\), the signal \(\mathbf{x}_i\) is treated as a zero vector. Note that only \(K_a\) out of the \(2^{B_p}\) indicators, \(b_i\)'s, are non-zero, and these non-zero indicators are detected jointly with the data.
 
The energy-per-bit, \( E_b/N_0 \), where \( N_0 = 2\sigma^2 \), is given by:

\begin{equation}
\frac{E_b}{N_0} = \frac{n_c}{2B\sigma^2}.
\end{equation}
The system's performance is evaluated using the per-user probability of error (PUPE), defined as:

\begin{equation}
P_e = \frac{1}{K_a} \sum_{\mathbf{u} \in \mathcal{A}} \mathbb{P}\left(\mathbf{u} \notin \mathcal{L}(\mathbf{r})\right),
\end{equation}
where \( \mathcal{L}(\mathbf{r}) \) denotes the list of recovered messages with a size of at most \( K_a \), and \( \mathcal{A} \) is the list of active UE's messages.
	\subsubsection{Iterative Decoding}
In this section, the receiver performs iterative decoding to jointly detect the pattern and data. An overview of the iterative decoding process is provided below:
\begin{itemize}
		\item[-]Firstly, via the a priori statistics information, the Gaussian approximation multi-user detector (MUD) computes the soft log-likelihood ratio (LLR) information $L^{m}_{i_k,j}$ of $k$-th non-zero element in the $i$-th pattern at the $j$-th channel use, i.e., $k\in \left\{1,2,\ldots,|\mathcal{D}|\right\},~j\in \left\{1,2,\ldots,n_p\right\}$, where the set $\mathcal{D}$ contains the potential indices of nonzero elements at a given channel use.
		\item[-]Subsequently, the LLRs serve as input for the channel decoder. The channel decoder produces extrinsic LLRs \( L^d_{i_k,j} \), where \( i \in \{1,2,\ldots,n_c\} \), and updates the a priori statistics using a posteriori probability information.
		\item[-]The above two steps iteratively exchange the LLRs \(L^m_{i_k,j} \rightleftharpoons L^d_{i_k,j}\), allowing the system to extract distinct patterns across the \(n_p\) channel uses and ultimately detect both the patterns and the data.
	\end{itemize} 
	
The received signal for the $j$-th channel use can be written as
	\begin{equation}\label{eq:4}
		r_j = \sum_{i=1}^{2^{B_p}}b_ip_ix_{i,j}+n_j
		= b_ip_ix_{i,j}+\underbrace{\sum_{l\neq i}^{2^{B_p}}b_lp_lx_{l,j}+n_j}_{\xi_{i,j}},
	\end{equation}
where $\xi_{i,j}$ is the interference-plus-noise term and $n_j\sim \mathcal{N}(0,\sigma^2)$. Let $\mathrm{E}\left(\xi_{i,j}\right)$ and $\mathrm{Var}\left(\xi_{i,j}\right)$ denote the expected mean and variance. The conditional probability density function (PDF) of \( r_j \), given \( b_i \) and \( x_{i,j} \), can be written as:
	\begin{equation}\label{eq:5}
		\begin{aligned}
			&p\left(r_{j}|b_{i},x_{i,j}\right)\\
			&=\frac{1}{\sqrt{2\pi\mathrm{Var}\left(\xi_{i,j}\right) }}\exp\left(-\frac{\left(r_{j}-\left(b_ip_ix_{i,j}+ \mathrm{E}\left(\xi_{i,j}\right) \right)\right)^{2}}{2 \mathrm{Var}\left(\xi_{i,j}\right) }\right).
		\end{aligned}
	\end{equation}
During each iteration, the patterns deemed as active indicate the amount of symbols to be detected during each channel uses, i.e., only $|\mathcal{D}|$ active $x_{i,j}$ at single channel use. Therefore, we use $i_k \in \{1,2,\ldots,K=|\mathcal{D}|\}$ to index the $k$-th element in the $i$-th pattern and use $x_{i_k,j}$ to denote the $i_k$-th element to be detected at $j$-th channel use. 
	
The decoding begins with the initialization of \( p\left(x_{i_k,j}=\pm 1\right) = \frac{1}{2} \), \( \mathrm{E}\left(\xi_{i,j}\right) = 0 \), and \( \mathrm{Var}\left(\xi_{i,j}\right) = \bar{K}_a + \sigma^2 \), which will be updated in subsequent iterations. The soft probabilistic information of the activity indicator \( L(b_i) \) can be estimated by combining the statistical information of the \( n_c \) non-zero elements across \( n_p \) channel uses:
\begin{equation}\label{eq:6}
	\begin{aligned}
		&L(b_i) = \log \frac{p(\mathbf{r}|b_i=1)}{p(\mathbf{r}|b_i=0)}, \\
		&= \sum_{j=1}^{n_c} \log \frac{\sum\limits_{x_{i_k,j}=\pm 1} p\left(x_{i_k,j}\right) p\left(r_j|b_i=1,x_{i_k,j}\right)}{p\left(r_j|b_i=0,x_{i_k,j}=0\right)},
	\end{aligned}
\end{equation}
where the magnitude \( |L(b_i)| \) indicates the strength of inactivity/activity, and the sign (negativity/positivity) indicates inactivity/activity, respectively. For the $k$-th active element in the $i$-th pattern at $j$-th channel use, $p(x_{i_k,j}=\pm 1)$ can be calculated by LLR information $L^m_{i_k,j}$:
	\begin{equation}\label{eq:7}
		\begin{aligned}
		p(x_{i_k,j}&=-1)=\frac{p(r_j|x_{i_k,j}=-1)}{p(r_j|x_{i_k,j}=-1)+p(r_j|x_{i_k,j}=+1)},
			\\&=\frac{1}{1+\frac{p(r_j|x_{i_k,j}=+1)}{p(r_j|x_{i_k,j}=-1)}}
			=\frac{1}{1+\exp\left(L^m_{i_k,j}\right)}.
		\end{aligned}
	\end{equation}
	\begin{equation}\label{eq:8}
		p(x_{i_k,j}=+1)=1-p(x_{i_k,j}=-1)=\frac{\exp\left(L^m_{i_k,j}\right)}{1+\exp\left(L^m_{i_k,j}\right)},
	\end{equation}
	where $L^m_{i,j}=\log{\frac{p(r_j|x_{i_k,j}=+1)}{p(r_j|x_{i_k,j}=-1)}}$ initializes the demodulation LLRs, and \eqref{eq:7} and \eqref{eq:8} are used to calculate $L(b_i)$ in \eqref{eq:6}. 
	
	With the aforementioned definitions, the iterative data decoding procedures $L^m_{i_k,j}\rightleftharpoons L^d_{i_k,j}$ are explained step by step:
	\begin{itemize}
		\item[]\textbf{Step 1: Update Signal and Interference Statistics}
		\begin{subequations}\label{eq:9}
			\begin{align}
				\mathrm{E}\left(r_j\right)&=\sum_{i_k=1}^{K}p_i\mathrm{E}\left(x_{i_k,j}\right),\label{eq:9.1}\\
				\mathrm{Var}\left(r_j\right)&=\sum_{i_k=1}^{K}p_i^2\mathrm{Var}\left(x_{i_k,j}\right)+\sigma^2,\label{eq:9.2}\\
				\mathrm{E}\left(\xi_{i_k,j}\right)&=\mathrm{E}\left(r_j\right)-p_i\mathrm{E}\left(x_{i_k,j}\right),\label{eq:9.3}\\
				\mathrm{Var}\left(\xi_{i_k,j}\right)&=\mathrm{Var}(r_j)-p_i^2\mathrm{Var}(x_{i_k,j}).\label{eq:9.4}
			\end{align}
		\end{subequations}
\item[] \textbf{Step 2: A Priori LLR Calculation }
		\begin{equation}\label{eq:10}
			L^m_{i_k,j} = 2p_i\cdot \frac{r_j-E\left(\xi_{i_k,j}\right)}{\mathrm{Var}\left(\xi_{i_k,j}\right)}.
		\end{equation}
\item[]\textbf{Step 3: Feed $L^m_{i_k,j}$ into the Decoder}\footnote{Channel codes such as LDPC and RA codes can generate soft LLRs as output. However, channel codes such as Polar codes are not suitable for generating LLRs directly as output.}\textbf{to Update}\footnote{For performance comparison with and without the subtraction of $L_{i_k,j}^m$, please refer to \cite[Fig. 3]{IDMA}. The repeated use of the MUD output as the prior for the next round of demodulation causes performance degradation.} $L^d_{i_k,j}$.

		\begin{equation}\label{eq:11}
			L^d_{i_k,j}=L^d_{i_k,j}-L^m_{i_k,j}.
		\end{equation}
		\item[]\textbf{Step 4: Update $\mathrm{E}\left(x_{i_k,j}\right)$ and $\mathrm{Var}\left(x_{i_k,j}\right)$}
		\begin{subequations}\label{eq:12}
			\begin{align}
				\mathrm{E}\left(x_{i_k,j}\right)&=\tanh\left(L^d_{i_k,j}/2\right),\label{eq:12.1}\\
				\mathrm{Var}\left(x_{i_k,j}\right)&=1-\left(\mathrm{E}\left(x_{i_k,j}\right)\right)^2.\label{eq:12.2}
			\end{align}
		\end{subequations}
	\end{itemize}    
The activity information $L(b_i)$ is updated by inserting the prior probabilities $p(x_{i_k,j}=\pm 1)$ from \eqref{eq:7} and \eqref{eq:8} into \eqref{eq:6}. For decoded messages that pass the channel code parity check with $L(b_i)>0$, the receiver performs SIC by subtracting the corresponding signals of the qualified messages from the received signal. The proposed receiver framework design is shown in Fig.~\ref{fig:GMAC_Decoder} where the inner loop refers to procedures from \eqref{eq:5} to \eqref{eq:12} and the outer loop refers SIC procedure.
	\section{TIN-SIC Power Division}\label{TIN-SIC}
In this section, we describe how to allocate power to different users through the appropriate selection of \( p_i \). The so-called optimal power allocation refers to a scenario where all \( K_a \) users are divided into \( m \) groups, each containing \( K_1, K_2, \ldots, K_m \) users, with corresponding power levels \( P_1 \leq P_2 \leq \ldots \leq P_m \). However, due to the random access nature of the proposed scheme, where no cooperation exists between users and the BS, assigning a specific power level to each active user is not feasible. Instead, we partition the columns of the pattern codebook into \( m \) groups, each associated with a different power ratio. This ensures that each user is randomly assigned a power level based on their selected column from the pattern codebook.

According to \cite{PD_Mohamod}, selecting an equal number of users in each group represents the optimal case, i.e., \( K_1 = \ldots = K_m = \frac{K_a}{m} \). Additionally, the optimized number of groups is determined by solving the following problem:
\begin{equation}\label{eq:15}
		\hat{m}=\min_m\left(1+\frac{K_a}{m}\gamma_0\right)^m,
\end{equation}
Where \( \gamma_0 \) is the signal-to-noise ratio of each user required to achieve the target error in the case of \( \frac{K_a}{m} \) users. In \cite{PD_Mohamod}, \( \gamma_0 \) is obtained through extensive simulations. However, we propose an analytical approach to estimate it by considering the TIN-SIC system.

\begin{theorem}

    \label{theorem1}
Consider a URA system with \( K_0 \) users, each transmitting \( B_0 \) bits of information over \( n_0 \) channel uses. Let \( \gamma \) be the ratio of each user's average per-channel-use power to the noise variance. An approximation of the PUPE for the TIN-SIC scheme is given by:
\begin{align}
    \epsilon\left(\gamma,K_0,n_0,B_0\right)\approx\sum_{k=1}^{K_0}\dfrac{\bar{K}_k}{K_0}\mathbb{P}(\xi_{n,k})\prod_{j=1}^{k-1}(1-\mathbb{P}(\xi_{n,j})),
\end{align}
where 
    $\mathbb{P}(\xi_{n,k})=Q\left({(C-R_c)}/{\sqrt{V_{dis}/n_0}}\right)^{\bar{K}_k}$, $\bar{K}_k=K_0-k+1$, $C=0.5\log_2(1+\alpha_k)$, $V_{dis}=\dfrac{\alpha_k}{2}\dfrac{\alpha_k+2}{(\alpha_k+1)^2}\log_e^2(2)$, and $R_c=B_0/n_0$, and $\alpha_k=\dfrac{\gamma}{(1+(K_0-k)\gamma)}$. 
\end{theorem}

\begin{proof}
    See Appendix \ref{appendixA}.
\end{proof}
Using Theorem~\ref{theorem1}, we can find $\gamma_0$ by solving the following problem 
\begin{align}\label{eq:13b}
    \gamma_0=\min_{\gamma}\left|\epsilon\left(\gamma,K_a/m,n,B\right)-\epsilon_{t}\right|,
\end{align}
Where \( \epsilon_t \) is the target PUPE. By substituting \( \gamma_0 \) into \eqref{eq:15}, the optimal number of power groups is obtained. According to \cite{PD_Mohamod}, the transmit power of the \( j \)-th power group is selected as \( P_j = P_1 \frac{\sigma_j^2}{\sigma_1^2} \), where \( \sigma_j^2 = \sigma^2 + \sum_{i=1}^{j-1} P_i \frac{K_a}{m} \). Thus, to control the system's performance, it is sufficient to control \( P_1 \), as the transmit powers of all other groups depend on \( P_1 \). Finally, after selecting \( P_j \)'s, we scale them to satisfy the unit power constraint as in \eqref{eq:1}, i.e., \( \frac{p_1^2 + p_2^2 + \ldots + p_m^2}{m} = 1 \), which gives the power ratio for the \( j \)-th group as:
	\begin{equation}\label{eq:18}
		p_j=\sqrt{\frac{mP_j}{P_1+P_2+\ldots+P_m}}.
	\end{equation}
	
	\section{Numerical Results}\label{numerical}
In this section, the proposed method is compared with state-of-the-art schemes that serve as benchmarks. As is common, all schemes share the same spectral efficiency of \( \frac{B}{n} = \frac{100}{30,000} \approx 0.00333 \), and the target PUPE is 0.05, which is a typical setting for GMAC URA. The benchmarks include various categories of basic GMAC URA codes from recent years: The achievability bound \cite{Polyanskiy}; CCS \cite{CCS} with different inner/outer codes such as cyclic redundancy check-aided block markov superposition transmission (CRC-BMST) \cite{CRC-BMST}, dynamic CCS \cite{dynamic_CCS}; RS with different channel codes and MUD methods such as RS with Polar code \cite{RS-Polar}, LDPC code \cite{RS-LDPC} and sparse Kronecker product (SKP) \cite{SKP} with tail-biting convolutional code; Another kind is sparse code category including ODMA with RA code \cite{ODMA_RA}, None-Binary-LDPC code \cite{ODMA_NB_LDPC}. In terms of this work, new-radio LDPC (NR-LDPC) with $1/3$ code rate is adopted as channel code. The transmission is divided into $J=16$ chunks and the codebook is fixed to $2^{13}$ and $B_c=44$. The collision error is controlled trivial compared with target PUPE and will be automatically accounted into PUPE if false restoration takes place due to collision. 
	\begin{figure}[htp]
		\centering
		\includegraphics[width=3.3in]{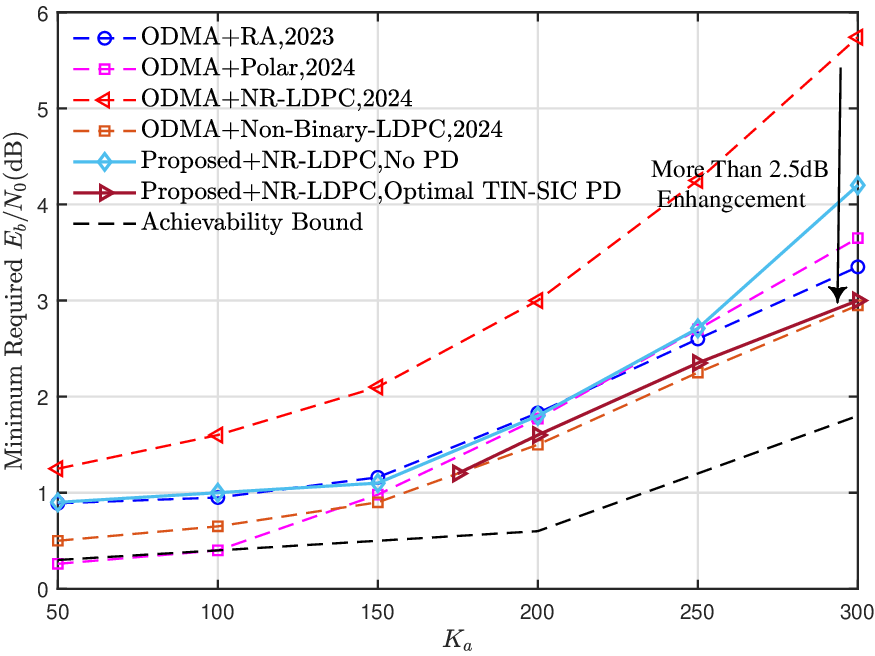}
\caption{Comparison of the minimum-required \( E_b/N_0 \) (dB) between the proposed scheme with/without TIN-SIC PD power division (PD), and ODMA-based schemes using NR-LDPC and polar code \cite{ODMA_Polar}, RA code \cite{ODMA_RA}, and non-binary LDPC \cite{ODMA_NB_LDPC}.}\label{fig:GMAC_ODMAs}
	\end{figure}
	
	\begin{figure}[htp]
		\centering
		\includegraphics[width=3.3in]{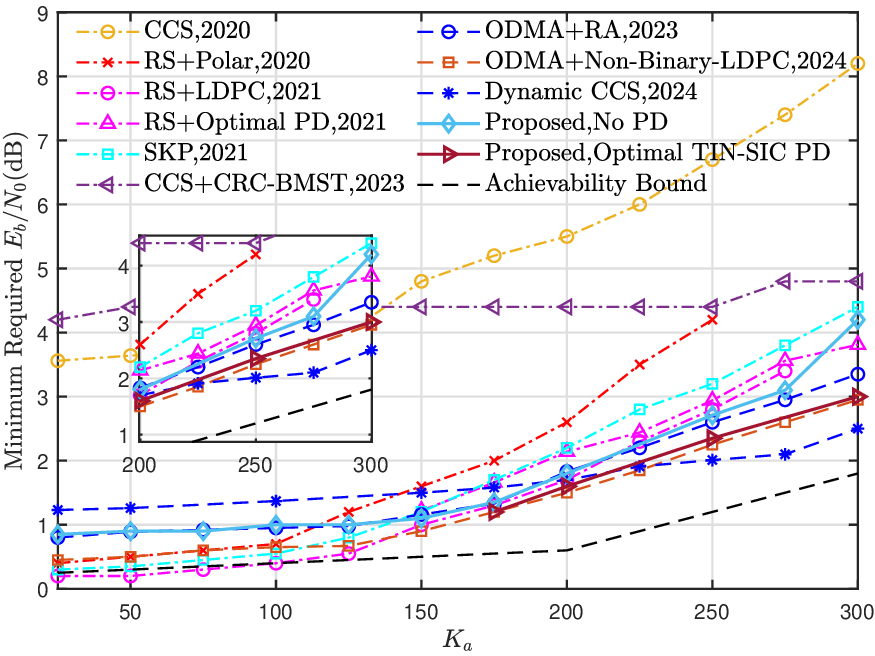}
		\caption{Performance of the minimum required energy-per-bit for a 0.05 PUPE. Chronologically, the benchmarks include achievability bound\cite{Polyanskiy}, CCS\cite{CCS}, RS-Polar\cite{RS-Polar}, RS-LDPC\cite{RS-LDPC}, RS+optimal power division\cite{PD_Mohamod}, SKP\cite{SKP}, CRC-BMST\cite{CRC-BMST}, ODMA+RA\cite{ODMA_RA}, ODMA+non-binary LDPC\cite{ODMA_NB_LDPC}, dynamic CCS\cite{dynamic_CCS}.}\label{fig:GMAC_Capacity}
	\end{figure}
	
%
		
		In Fig.~\ref{fig:GMAC_ODMAs} and Fig.~\ref{fig:GMAC_Capacity}, the capacity performance in terms of the minimum-required energy-per-bit $E_b/N_0$(dB) among various GMAC URA schemes is illustrated and compared. In Fig.~\ref{fig:GMAC_ODMAs}, the capacity performance of the proposed scheme with/without TIN-SIC PD is compared with ODMA structure with different channel codes. Particularly, the original ODMA with NR-LDPC has a capacity performance 2.75 dB worse than ODMA with non-binary LDPC at $K_a=300$. While the proposed scheme adopts the same channel code (NR-LDPC), it shows very similar performance to ODMA with Non-Binary LDPC with optimal TIN-SIC PD. This outcome demonstrates the superiority of the proposed scheme especially in terms of the effective TIN-SIC power division method. Moreover, the capacity performance with other types of basic codes for GMAC URA is illustrated in Fig.~\ref{fig:GMAC_Capacity}. Overall, the proposed scheme has better capacity performance than many existing code designs. If a more powerful channel code is equipped, the performance can be further improved.
		\section{Conclusion}\label{conclusion}
In this work, an analytical power division method has been developed and integrated into the sparse code framework within the URA GMAC. Compared to a similar scheme with an identical channel code, a 2.75 dB capacity gain can be observed with the proposed transmission structure along with power division. While the proposed scheme significantly improves the overall performance of sparse code-based URA schemes, achieving the theoretical limits still requires breakthroughs in channel code design with advanced decoding techniques.
\begin{appendices}
\section{Proof of Theorem \ref{theorem1}}
    \label{appendixA}
Here, we calculate the PUPE for the TIN-SIC scheme. In TIN-SIC, we assume that at each iteration, at least one user is decoded, and the contribution of that user is removed from the received signal using SIC. Therefore, in the \( k \)-th iteration, there remain \( K_0-k+1 \) users that need to be decoded. We also assume that if no user is decoded in an iteration, the algorithm is stopped. Let \( \xi_{s,k} \) be the event that the iterative algorithm is stopped in the \( k \)-th iteration, and \( \xi_{n,k} \) denote the event that none of the users are decoded in the \( k \)-th iteration. \( \xi_{s,k} \) occurs if the algorithm is not stopped in the previous \( k-1 \) iterations, and it is stopped in the current iteration. Assuming the decoding process in different iterations to be independent of each other, the probability that \( \xi_{s,k} \) occurs is given by $   \mathbb{P}(\xi_{s,k})\approx\mathbb{P}(\xi_{n,k})\prod_{j=1}^{k-1}(1-\mathbb{P}(\xi_{n,j}))$.
Besides, $\mathbb{P}(\xi_{n,k})$ can be approximated as \cite{Block_Error}
\begin{align}
    \mathbb{P}(\xi_{n,k})\approx Q\left(\dfrac{C-R_c}{\sqrt{V_{dis}/n_0}}\right)^{(K_0-k+1)},
\end{align}
where $C$, $V_{dis}$, and $R_c$ are as defined in the statement of the Theorem \ref{theorem1}. Moreover, if $\xi_{s,k}$ happens, the PUPE is calculated as ${(K_0-k+1)}/{K_0}$. Therefore, the total PUPE of the system is calculated as
\begin{align}
    \epsilon=\sum_{k=1}^{K_0}\mathbb{P}(\xi_{s,k})\dfrac{K_0-k+1}{K_0}.
\end{align}
\end{appendices}


		
		%

		\vfill

\begin{thebibliography}{1}
			\bibliographystyle{IEEEtran}
			

			\bibitem{Polyanskiy}
			Y. Polyanskiy, ``A perspective on massive random-access,'' in {\em Proc. IEEE IEEE Int. Symp. Inf. Theory (ISIT)}, pp. 2523--2527, 25-30 Jun. 2017, Aachen, Germany.
			
			\bibitem{acdemia and industrial1}
			P. Agostini {\em et al.}, ``Enhancements to the 5G-NR 2-step RACH: an unsourced multiple access perspective,'' in {\em Proc. Conf. Standards for Commun. and Network. (CSCN)}, Belgrade, Serbia, 2024, pp. 32-35.
			\bibitem{acdemia and industrial2}
			Y. Yuan {\em et al.}, ``Unsourced sparse multiple access: A promising transmission scheme for 6G massive communication,'' {\em IEEE Commun. Mag.}, \url{doi: 10.1109/MCOM.002.2400453}, 2025.
			
			
			\bibitem{PD_Mohamod}
			M. J. Ahmadi and T. M. Duman, ``Random spreading for unsourced MAC with power diversity,'' {\em IEEE Commun. Lett}, vol. 25, no. 12, pp. 3995-3999, Dec. 2021.
			
			\bibitem{RS-Polar}
			A. K. Pradhan, {\em et al.}, ``Polar coding and random spreading for unsourced multiple access,'' in {\em Proc. IEEE Int. Comf. Commun.}, Dublin, Ireland, 2020, pp. 1-6.
			
			\bibitem{RS-LDPC}
			A. K. Pradhan, {\em et al.}, ``LDPC codes with soft interference cancellation for uncoordinated unsourced multiple access,'' in {\em Proc. IEEE Int. Comf. Commun.}, Montreal, QC, Canada, 2021, pp. 1-6.
		
			
			\bibitem{SKP}
			Z. Han, {\em et al.}, ``Sparse kronecker-product coding for unsourced multiple access,'' {\em IEEE Wireless Commun. Lett.}, vol. 10, no. 10, pp. 2274-2278, Oct. 2021.
			
			\bibitem{CCS}
			V. K. Amalladinne, {\em et al.}, ``A coded compressed sensing scheme for unsourced multiple access,'' {\em IEEE Trans. Inf. Theory}, vol. 66, no. 10, pp. 6509-6533, Oct. 2020.
			\bibitem{CCS2}
			V. K. Amalladinne, {\em et al.}, ``Unsourced random Access With Coded Compressed Sensing: Integrating AMP and Belief Propagation,'' {\em IEEE Trans. Inf. Theory}, vol. 68, no. 4, pp. 2384-2409, April 2022.
			\bibitem{CRC-BMST}
			H. Cao, {\em et al.}, ``CRC-aided sparse regression codes for unsourced random access,'' {\em IEEE Commun. Lett}, vol. 27, no. 8, pp. 1944-1948, Aug. 2023.
			\bibitem{dynamic_CCS}
			E. Nassaji and D. Truhachev, ``Dynamic compressed sensing approach for unsourced random access,'' {\em IEEE Commun. Lett}, vol. 28, no. 7, pp. 1644-1648, July 2024.
			
			\bibitem{ODMA_RA}
			J. Yan, {\em et al.}, ``ODMA transmission and joint pattern and data recovery for unsourced multiple access,'' {\em IEEE Wireless Commun. Lett.}, vol. 12, no. 7, pp. 1224-1228, July 2023.
			\bibitem{ODMA_NB_LDPC}
			J. Yan, {\em et al.}, ``Enhanced ODMA with channel code design and pattern collision resolution for unsourced multiple access,'' in {\em Proc. IEEE Int. Symp. Inf. Theory (ISIT)}, Athens, Greece, 2024, pp. 3201-3206.
			\bibitem{ODMA_Polar}
			M. Ozates, {\em et al.}, ``Unsourced random access using ODMA and polar codes,'' {\em IEEE Wireless Commun. Lett.}, vol. 13, no. 4, pp. 1044-1047, April 2024.

            
\bibitem{andreev2020polar} K. Andreev, E. Marshakov and A. Frolov, ``A polar code based TIN-SIC scheme for the unsourced random access in the quasi-static fading MAC,'' in \emph{Proc. IEEE Int. Symp. Inf. Theory (ISIT)}, Los Angeles, USA, June 2020, pp. 3019--3024.
\bibitem{Ahmadi2023Unsourced}  M. J. Ahmadi, M. Kazemi, and T. M. Duman, ``Unsourced random access with a massive MIMO receiver using multiple stages of orthogonal pilots: MIMO and single-antenna structures,'' \emph{IEEE Trans. Wireless Commun.}, vol. 23, no. 2, pp. 1343--1355, Feb. 2024.
			
			\bibitem{Block_Error}
			Y. Polyanskiy, {\em et al.}, ``Channel coding rate in the finite blocklength regime,'' {\em IEEE Trans. Inf. Theory}, vol. 56, no. 5, pp. 2307–2359, May 2010.
			\bibitem{IDMA}
			Li Ping, {\em et al.}, ``Interleave division multiple-access,'' {\em IEEE Trans. Wireless Commun.}, vol. 5, no. 4, pp. 938-947, April 2006.
			
			
			
			
			
			
		\end{thebibliography}
	\end{document}